\pdfoutput=1

\documentclass[paper=a4,fontsize=10pt,DIV12,abstract=true]{scrartcl}

\usepackage{etoolbox}
\usepackage[utf8]{inputenc}
\usepackage[T1]{fontenc}
\usepackage{lmodern}
\usepackage[activate,final]{microtype}
\usepackage{flafter}

\newcommand{\miniheading}[1]{\medskip\noindent\textbf{#1}}
\typearea[current]{last}
\addtokomafont{sectioning}{\rmfamily}
\addtokomafont{descriptionlabel}{\rmfamily}
\usepackage[USenglish]{babel}
\usepackage{amsmath,amssymb,amsthm}
\newtheorem{theorem}{Theorem}[section]
\newtheorem{lemma}[theorem]{Lemma}
\newtheorem{corollary}[theorem]{Corollary}
\theoremstyle{definition}

\newtheorem{claim}[theorem]{Claim}
\newtheorem{plainproblem}[theorem]{Problem}

\usepackage{enumitem}
\newenvironment{problem}[2][]{\begin{plainproblem}[#1] #2%
  \begin{description}[nosep,leftmargin=1.90cm,%
    style=sameline,align=right,format=\normalfont\emph,
    beginpenalty=10000,midpenalty=5000]%
}{%
  \end{description}\end{plainproblem}%
}

\DeclareMathOperator{\Aut}{Aut}
\DeclareMathOperator{\supp}{supp}
\DeclareMathOperator{\Sym}{Sym}
\DeclareMathOperator{\Alt}{Alt}

\usepackage{mathtools}
\DeclarePairedDelimiter{\pair}{\langle}{\rangle}
\DeclarePairedDelimiter{\card}{\lvert}{\rvert}

\makeatletter
\DeclarePairedDelimiter{\SK@setone}{\lbrace}{\rbrace}
\DeclarePairedDelimiterX{\SK@settwo}[2]{\lbrace}{\rbrace}{#1\:\delimsize\vert\:#2}
\newcommand{\set}{\@ifstar{\SK@set{*}}{\SK@oset}}
\newcommand{\SK@oset}[1][]{\ifstrempty{#1}{\SK@set{}}{\SK@set{[#1]}}}
\newcommand{\SK@set}[2]{\@gifnextchar\bgroup{\SK@@set{#1}{#2}}{\SK@setone#1{#2}}}
\newcommand{\SK@@set}[3]{\ifstrempty{#1}{%
    \SK@settwo{#2}{#3}%
  }{%
    \SK@settwo#1{#2}{\begin{array}{@{}l@{}}#3\end{array}}%
  }}
\def\@gifnextchar#1#2#3{\let\@tempe#1\def\@tempa{#2}\def\@tempb{#3}%
  \futurelet\@tempc\@gifnch}
\def\@gifnch{\ifx\@tempc\@sptoken\let\@tempd\@tempb%
  \else\ifx\@tempc\@tempe\let\@tempd\@tempa\else\let\@tempd\@tempb\fi\fi\@tempd}
\makeatother

\usepackage{xspace}
\newcommand{\cclass}[1]{\textsf{\upshape #1}}
\newcommand{\FPT}{\cclass{FPT}\xspace}
\newcommand{\WP}{\cclass{W[P]}\xspace}
\newcommand{\NP}{\cclass{NP}\xspace}
\newcommand{\Pcl}{\cclass{P}\xspace}
\newcommand{\wone}{\cclass{W[1]}\xspace}
\newcommand{\wtwo}{\cclass{W[2]}\xspace}

\newcommand{\gclass}[1]{\textsc{#1}}

\newcommand{\ldiscrete}{\gclass{Discrete$[\ell]$}}
\newcommand{\prob}[1]{\textsc{#1}}
\newcommand{\WMCS}{\prob{Weighted Monotone Circuit Satisfiability}\xspace}
\newcommand{\IndC}{\prob{$k$-$\mathcal{C}$}\xspace}
\newcommand{\IndDiscrete}{\prob{$k$-Discrete}\xspace}
\newcommand{\IndLDiscrete}{\prob{$k$-\ldiscrete}}
\newcommand{\IndRefinable}{\prob{$k$-Refinable}\xspace}
\newcommand{\kdiscrete}{\prob{$k$-Discrete}\xspace}
\newcommand{\nkdiscrete}{\prob{$(n-k)$-Discrete}\xspace}
\newcommand{\DomSet}{\prob{Dominating Set}\xspace}

\newcommand{\msat}{\prob{Mini-3SAT}}

\newcommand{\kcoldeg}{\prob{$k$-Color-Valence}\xspace}

\newcommand{\F}{\mathbb{F}}

\newcommand{\CFI}{\textrm{\upshape CFI}}
\newcommand{\IMP}{\textrm{\upshape IMP}}

\newcommand{\kfixing}{\prob{$k$-Rigid}\xspace}
\newcommand{\kbase}{\prob{$k$-Base-Size}\xspace}

\newcommand{\nkfixing}{\prob{$(n-k)$-Rigid}\xspace}
\newcommand{\nkbase}{\prob{$(n-k)$-Base-Size}\xspace}

\newcommand{\setcover}{\prob{Set-Cover}\xspace}

\newcommand{\mini}{\cclass{MINI[1]}\xspace}

\newcommand{\stabcol}{\mathcal{C}}
\newcommand{\neigh}{\mathcal{N}}

\usepackage{tikz}
\usetikzlibrary{backgrounds}

\usepackage[final,
            bookmarksopen=true,
            pdfstartview=FitH,
            pdfprintscaling=None,
            ]{hyperref}
\hypersetup{
  unicode=true,
  pdftitle={The Parameterized Complexity of Fixing Number and
    Vertex Individualization in Graphs},
  pdfauthor={Vikraman Arvind, Frank Fuhlbrück, Johannes
    Köbler, Sebastian Kuhnert, and Gaurav Rattan},
}

\title{The Parameterized Complexity of Fixing Number and Vertex
  Individualization\\ in Graphs\footnote{An abridged version of this
    article appears in the proceedings of MFCS~2016. This work was
    supported by the Alexander von Humboldt Foundation research
    group linkage program. The third and fourth authors are
    supported by DFG grant~KO~1053/7-2.}}
\newcommand{\authornote}[1][1]{\textsuperscript{#1}}
\author{V.~Arvind\authornote[1]\and
  Frank~Fuhlbrück\authornote[2]\and
  Johannes~Köbler\authornote[2]\and
  Sebastian~Kuhnert\authornote[2]\and
  Gaurav~Rattan\authornote[1]}
\date{June 14, 2016}
\begin{document}
\maketitle
\footnotetext[1]{The Institute of Mathematical Sciences, Chennai, India;
  \texttt{\{arvind,grattan\}@imsc.res.in}}
\footnotetext[2]{Humboldt-Universität zu Berlin, Germany;
  \texttt{\{fuhlbfra,koebler,kuhnert\}@informatik.hu-berlin.de}}
\begin{abstract}
  In this paper we study the complexity of the following problems:
\begin{enumerate}
\item Given a colored graph $X=(V,E,c)$, compute a minimum cardinality set of
  vertices $S\subset V$ such that no nontrivial automorphism of~$X$
  \emph{fixes} all vertices in~$S$. A closely related problem is
  computing a minimum base~$S$ for a permutation group $G\le S_n$
  given by generators, i.e., a minimum cardinality subset
  $S\subset [n]$ such that no nontrivial permutation in~$G$ fixes all
  elements of~$S$. Our focus is mainly on the \emph{parameterized
    complexity} of these problems. We show that when $k=\card{S}$ is
  treated as parameter, then both problems are $\mini$-hard. For the dual
  problems, where $k=n-\card{S}$ is the parameter, we give
  FPT~algorithms.
\item A notion closely related to fixing is called individualization.
  Individualization combined with the Weisfeiler-Leman procedure is a
  fundamental technique in algorithms for Graph Isomorphism.
  Motivated by the power of individualization, in the present paper we
  explore the complexity of individualization: what is the minimum
  number of vertices we need to individualize in a given graph such
  that color refinement ``succeeds'' on it. Here ``succeeds'' could
  have different interpretations, and we consider the following: It
  could mean the individualized graph becomes: (a)~discrete, (b)~amenable,
  (c)~compact, or (d)~refinable. In particular, we study
  the parameterized versions of these problems where the parameter is the
  number of vertices individualized. We show a dichotomy: For graphs
  with color classes of size at most~$3$ these problems can be solved
  in polynomial time (even in logspace), while
  starting from color class size~$4$ they become \WP-hard.
 \end{enumerate}
\end{abstract} 

\section{Introduction}

A permutation~$\pi$ on the vertex set~$V$ of a (vertex) colored graph
$X=(V,E,c)$ is an \emph{automorphism} if $\pi$~preserves edges and
colors. Uncolored graphs can be seen as the special case where all
vertices have the same color. The automorphisms of~$X$ form the
group~$\Aut(X)$, which is a subgroup of the symmetric group~$\Sym(V)$
of all permutations on~$V$.

A \emph{fixing set} for a colored graph $X=(V,E,c)$ is a subset~$S$ of vertices
such that there is no nontrivial automorphism of~$X$ that fixes every
vertex in~$S$. The \emph{fixing number} of~$X$ is the
cardinality of a smallest size fixing set of~$X$. This notion was
independently studied in~\cite{Boutin,Harary,Mohar}. A nice survey on
this and related topics is by Bailey and Cameron~\cite{BaileyCameron}.

In this paper, one of the problems of interest is the computational
complexity of computing the fixing number of graphs:
\begin{problem}{\kfixing}
 \item[Input:] A colored graph~$X$ and an integer~$k$
 \item[Parameter:] $k$
 \item[Question:] Is there a subset~$S$ of~$k$ vertices in~$V$
  such that there are no nontrivial automorphisms of~$X$ that fix
  each vertex of~$S$?
\end{problem}

There is a closely related problem that has received some attention.
Let $G\le S_n$ be a permutation group on~$[n]$. A \emph{base} of~$G$
is a subset $S\subset [n]$ such that no nontrivial permutation of~$G$
fixes each point in~$S$, i.e., the pointwise stabilizer subgroup
$G_{[S]}=\{g\in G\mid i^g=i~\forall~i\in S\}$ of~$G$
is the trivial subgroup~$\{1\}$.

\begin{problem}{\kbase}
\item[Input:] A generating set for a permutation group~$G$ on~$[n]$ and an
  integer~$k$
 \item[Parameter:] $k$
 \item[Question:] Is there a subset $S\subset [n]$ of size~$k$ such
   that no nontrivial permutation of~$G$ fixes each point in~$S$?
\end{problem}
Note that a graph~$X$ is in~\kfixing if and only if $\Aut(X)$~is
in~\kbase.

Computing a minimum cardinality base for $G\le S_n$ given by
generators is shown to be \NP-hard by Blaha~\cite{Blaha}. The same paper also
gives a polynomial-time $\log\log n$ factor approximation algorithm
for the problem, i.e., the algorithm outputs a base of size bounded by
$b(G)\log\log n$, where $b(G)$~denotes the optimal base size. We show
that this approximation factor cannot be improved unless $\Pcl=\NP$;
see Theorem~\ref{th:blaha-optimal}.

In this paper our focus is on the parameterized complexity of these
problems. Arvind has shown that \kbase is in~\FPT for transitive
groups and groups with constant orbit size~\cite{Arv13}, and raised
the question whether this can be extended to more general permutation
groups. We show that both \kfixing and \kbase are $\mini$-hard, even
when the automorphism group of the given graph~$X$ (resp., the given
group~$G$) is an elementary $2$-group; see Section~\ref{sec:fixed}.

We also consider the dual problems \nkfixing and \nkbase, which ask
whether the given graph or group have a fixing set or base that
consists of all but~$k$ vertices or points and $k$~is the parameter.
We show that these problems are fixed parameter tractable. More
precisely, we give an $k^{O(k^2)}+k\,n^{O(1)}$ time algorithm
for~\nkbase and an $k^{O(k^2)}n^{O(1)}$ time algorithm for~\nkfixing
in Section~\ref{sec:unfixed}.

\miniheading{Color refinement and individualization.}
A broader question that arises is in the context of the Graph
Isomorphism problem: Given two colored graphs $X=(V,E,c)$ and
$X'=(V',E',c')$ the problem is to decide if they
are \emph{isomorphic}, i.e., whether there is a bijection $\pi\colon
V\to V'$ such that for all $x\in V$, $c'(x^\pi)=c(x)$ and for all
$x,y\in V$, $(x,y)\in E$ if and only if $(x^\pi,y^\pi)\in E'$.

Color refinement is a classical heuristic for Graph Isomorphism, and
in combination with other tools (group-theoretic/combinatorial) it has
proven successful in Graph Isomorphism algorithms (e.g.\ in the two
most important papers in the area~\cite{BL83,B15}). The basic color
refinement procedure works as follows on a given colored graph
$X=(V,E,c)$. Initially each vertex has the color given by~$c$. The
refinement step is to color each vertex by the tuple of its own color
followed by the colors of its neighbors (in color-sorted order). The
refinement procedure continues until the color classes become stable.
If the multisets of colors are different for two graphs $X$~and~$X'$,
we can conclude that they are not isomorphic. Otherwise, more
processing needs to be done to decide if the input graphs are
isomorphic. One important approach in this area is to combine
\emph{individualization} of vertices with color refinement: Given a
graph $X=(V,E)$ and $k$~vertices $v_1,v_2,\ldots,v_k\in V$, first
these $k$~vertices are assigned distinct colors $c_1,c_2,\ldots,c_k$,
respectively. Then, with this as initial coloring, the color
refinement procedure is carried out as before. Individualization is
used both in the algorithms with the best worst case
complexity~\cite{BL83,B15} and in practical isomorphism
solvers~\cite{MP14}. Note that individualizing a vertex~$v$ results in
fixing~$v$, as every automorphism must preserve the unique color
of~$v$.

In~\cite{AKRV} we have examined several classes of colored graphs in
connection with the color refinement procedure. They form a hierarchy:
\begin{equation}
  \label{eq:cr-hier}
  \gclass{Discrete}\subsetneq
  \gclass{Amenable}\subsetneq
  \gclass{Compact}\subsetneq
  \gclass{Refinable}
\end{equation}
\begin{itemize}
 \item $X\in\gclass{Discrete}$ if running color refinement on~$X$
  results in singleton color classes.
 \item $X\in\gclass{Amenable}$ if for any~$X'$ that is non-isomorphic
  to~$X$, color refinement on $X$~and~$X'$ results in different stable
  colorings~\cite{AKRV}.
 \item $X\in\gclass{Compact}$ if every fractional automorphism of~$X$
  is a convex combination of automorphisms of~$X$~\cite{Tin91}. Here,
  automorphisms are viewed as permutation matrices that commute with
  the adjacency matrix~$A$ of~$X$, and fractional automorphisms are
  doubly stochastic matrices that commute with~$A$.
 \item $X\in\gclass{Refinable}$ if two vertices $u$~and~$v$ of~$X$
  receive the same color in the stable coloring if and only if there
  is an automorphism of~$X$ that maps~$u$ to~$v$~\cite{AKRV}.
\end{itemize}
For these graph classes, various efficient isomorphism and
automorphism algorithms are known. Motivated by the power of
individualization in relation to color refinement, we consider the
following type of problems.
\begin{problem}{\IndC (where $\mathcal{C}$ is a class of colored graphs)}
 \item[Input:] A colored graph~$X=(V,E,c)$ and an integer~$k$
 \item[Parameter:] $k$
 \item[Question:] Are there $k$~vertices of~$X$ so that
  individualizing them results in a graph in~$\mathcal{C}$?
\end{problem}

It turns out that for each class~$\mathcal{C}$ in the
hierarchy~\eqref{eq:cr-hier}, the problem~\IndC is \WP-hard, even when
the input graph has color class size at most~$4$. For color class size
at most~$3$ however, the problems become polynomial time solvable. For
the class $\ldiscrete$ of all colored graphs where $\ell$~rounds of
color refinement turn all color classes into singletons, we show that
$\IndLDiscrete$ is \wtwo-hard. These results are in
Section~\ref{sec:indiv}.

Additionally, we give an FPT algorithm for the dual problem
\nkdiscrete that asks whether there is a way to individualize all
but~$k$ vertices so that the input graph becomes discrete; see
Section~\ref{sec:nonindiv}.

\miniheading{Color valence.}
A beautiful observation due to Zemlaychenko~\cite{Zem}, that plays a
crucial role in~\cite{BL83}, concerns the color valence of a graph.
Given a colored graph $X=(V,E,c)$, the \emph{color degree}~$\deg_C(v)$
of a vertex~$v$ in a color class~$C=\set{v\in V}{c(v)=c_0}$ is the number of
neighbors of~$v$ in~$C$. The \emph{color co-degree} of~$v$ in~$C$ is
$\text{co-}\!\deg_C(v)=\card{C}-\deg_C(v)$. The \emph{color valence} of~$X$
is defined as $\max_{v,C} \min\{\deg_C(v),\text{co-}\!\deg_C(v)\}$.
Zemlyachenko has shown~\cite{Zem} that in any $n$-vertex graph
$X=(V,E)$ we can individualize $O(n/d)$~vertices so that the vertex
colored graph obtained after color refinement has color valence at
most~$d$. This gives rise to the following natural algorithmic
problem:

\begin{problem}{\kcoldeg}
\item[Input:] A colored graph~$X=(V,E,c)$ and two numbers $k$~and~$d$
 \item[Parameter:] $k$
 \item[Question:] Is there a set of $k$~vertices such that when these
   are individualized, the graph obtained after color refinement has
   color valence bounded by~$d$?
\end{problem}
We show that this problem is \WP-complete; see Corollary~\ref{cor:coldeg}.

\section{The number of fixed vertices as parameter}\label{sec:fixed}

In this section we show that the parameterized problems $\kfixing$ and
$\kbase$ are both $\mini$-hard.  The class~$\mini$ contains all 
parameterized problems that are FPT-reducible 
to~$\msat$. Both were defined in~\cite{CJ,DEFPR03}.

\begin{problem}[\cite{CJ,DEFPR03}]{\msat}
\item[Input:] A formula~$F$ in 3-CNF of size bounded by~$k\log n$ and
  the number~$n$ in unary
 \item[Parameter:] $k$
 \item[Question:] Is there a boolean assignment to the variables that
    satisfies the formula~$F$?
\end{problem}

It turns out that $\mini$~is contained in the class~$\wone$~\cite{DEFPR03}
and has a variety of complete problems in it. Moreover, it has been
linked to the exponential time hypothesis.
\begin{lemma}[\cite{CJ,DEFPR03}]
  If $\mini=\FPT$ then there is a $2^{o(n)}$ time algorithm for 3SAT.
\end{lemma}

\begin{theorem}\label{base-mini-hard}
  The problem $\kbase$ is $\mini$-hard, even for elementary $2$-groups.
\end{theorem}  

\begin{proof}
  It is easy to see that $\msat$ in which each variable occurs at 
  most~$3$ times is also $\mini$-complete, by modification of a 
  standard NP-completeness proof. This only increases the size by a
  constant factor. We can therefore assume that a given $\msat$
  instance has this property.
  
  We will give an FPT many-one reduction from $\msat$ to $\kbase$.
  Let $F=C_1\wedge C_2\wedge\cdots\wedge C_m$, and $n$~in unary, be a
  $\msat$ instance with variable occurrences bounded by~$3$. Since the
  size of~$F$ is bounded by~$k\log n$, we have $m\le k\log n$.  Let~$V$
  denote the set of distinct variables in~$F$. We also have
  $\card{V}\le k\log n$. We partition~$V$ as $V=\bigsqcup_{i=1}^k V_i$,
  where $\card{V_i}\le \log n$ for $1\le i\le k$. For each~$i$, the set
  $T_i=\{0,1\}^{V_i}$ consisting of all truth assignments to variables in~$V_i$
  has size $\card{T_i}\le n$. Define the universe
  $U=\{1,2,\ldots,m,m+1,\ldots,m+k\}$. For each truth assignment $a\in
  T_i$ we define the subset $S_{i,a}\subset U$ consisting of $m+i$
  along with all~$j$ such that $a$~satisfies~$C_j$, i.e.,
  \[
    S_{i,a}=\{m+i\}\cup \{j\mid C_j \textrm{ contains a literal that is true under }a\}.
  \]
  Clearly, since each variable occurs at most $3$~times in~$F$ and since
  $\card{a}=\card{V_i}\le\log n$, it
  follows that $\card{S_{i,a}}\le 1+3\log n$. The following claim is
  straightforward.

\begin{claim}\label{setcover}
  The collection of sets $\{S_{i,a}\mid 1\le i\le k, a\in T_i\}$ with
  universe~$U$ has a set cover of size~$k$ if and only if $F$~is
  satisfiable.
\end{claim}
We will now transform this special set cover instance into an instance
of $\kbase$. The group we shall consider is~$\F_2^{m+k}$, i.e., the
product of $m+k$ copies of the group on $\{0,1\}$ defined by addition
modulo~$2$. Treating each set~$S_{i,a}$ as a subset of the coordinates
$1,2,\ldots,m+k$, we can associate a copy of~$\F_2^{\card{S_{i,a}}}$
with it. Consider the set
$\Omega=\bigsqcup_{i,a}\F_2^{\card{S_{i,a}}}$. Note that
$\card{\Omega}=\sum_{i,a}2^{\card{S_{i,a}}}\le nk$. The
group~$\F_2^{m+k}$ acts faithfully on~$\Omega$ as follows. Given an
element $u\in\F_2^{m+k}$ and a point $v\in \F_2^{\card{S_{i,a}}}$,
let~$u_{i,a}$ denote the projection of~$u$ to the coordinates
in~$S_{i,a}$. Then $u$~maps~$v$ to $v\oplus u_{i,a}$. Thus,
$\F_2^{m+k}$~is a permutation group acting on~$\Omega$ given by the
standard basis of $m+k$ unit vectors as generating set. The following
straightforward claim completes the reduction.

\begin{claim}
  The group~$\F_2^{m+k}$ acting on~$\Omega$, as defined
  above, has a base of size~$k$ if and only if the set cover instance
  $(U, \{S_{i,a}\mid 1\le i\le k, a\in T_i\})$ has a set cover of size~$k$.
\end{claim}
To see the claim, observe that $V\subseteq\Omega$ is a base if and only if the
sets~$S_{i,a}$ with $V\cap\F_2^{\card{S_{i,a}}}\neq\emptyset$ form a set cover
for~$U$. Indeed, a point $p\in U$ is covered by these sets if and only if
all $u\in\F_2^{m+k}$ with $u_p=1$ move an element of~$V$.
\end{proof}

\begin{theorem}\label{fixing-mini-hard}
  The problem $\kfixing$ is $\mini$-hard, even for graphs whose automorphism groups
    are elementary $2$-groups.
\end{theorem}  
\begin{proof}
It suffices to encode the $\kbase$ instance constructed in the proof of
Theorem~\ref{base-mini-hard} as a $\kfixing$ instance $(X,k)$ with the
following properties. The graph~$X$ has $\card{\Omega}+2(m+k)$ vertices and
at most $\card{\Omega}(1+3\log n)$ edges.  Further, the above $\kbase$
instance has a base of size~$k$ if and only if the graph~$X$ has a
fixing set of size~$k$.

We explain the construction of~$X$. Let $l=m+k$.  The vertex set of~$X$
is $\Omega \cup I_1 \cup \cdots \cup I_l$ where each set
$I_j = \{a^0_j,a^1_j\}$ is a distinct color class of size~$2$.
The edge set of~$X$ is defined as follows. Let
$v=(b_1, \dots, b_p)\in\F_2^{\card{S_{i,a}}}$ be a vertex in~$\Omega$ and
let $S_{i,a}=\{i_1,i_2,\dots,i_p\}$ be the set of coordinates
occurring in~$v$. Then we connect~$v$ to the vertices~$a_{i_q}^{b_q}$,
for each $q=1, \dots, p$.  This finishes the construction of~$X$.

We claim a one-to-one correspondence between the permutation
group~$\F_2^{m+k}$ acting on~$\Omega$ and~$\Aut(X)$.  Indeed, any vector
$v = (b_1, \dots, b_l) \in \F_2^{m+k}$ can be associated with a unique
automorphism~$\sigma$ of~$X$ as follows. The automorphism~$\sigma$
flips the color class~$I_j$ if and only if $b_j=1$.
For a vertex $u\in\Omega$, define $\sigma(u)=v(u)$ using the action
of~$\F_2^{m+k}$ on~$\Omega$.
It is easy
to check that $\sigma$~respects the adjacencies inside~$X$. Note that the
action of an automorphism of~$X$ is determined by its action on
$I_1,\dots,I_l$, so this is a one-to-one correspondence.

Consequently, a set $J \subset \Omega$ is a base for the original
$\kbase$ instance if and only if $J$~is a fixing set for the graph~$X$.
We observe that we can always avoid fixing a vertex~$u$ inside
$I_1\cup \dots \cup I_l$ by instead fixing some neighbor of
$u\in\Omega$.  Therefore, the original $\kbase$ instance has a base of
size~$k$ if and only if the graph~$X$ has a fixing set of size~$k$.
\end{proof}  

We end this section with some consequences of our hardness proofs on
the approximability of the minimum base size of a group.  There is a
$\log\log n$ factor approximation algorithm due to Blaha~\cite{Blaha}
for the minimum base problem (in fact, a careful analysis yields a
$\ln \ln n$-factor approximation). In this connection we have an
interesting observation about the set cover problem instances that
arise in Theorem~\ref{base-mini-hard} (Claim~\ref{setcover}). A more
general version is the $B$-$\setcover$ problem: we are given a
collection of subsets of size at most~$B$ of some universe~$U$ and the
problem is to find a minimum size set cover. Trevisan~\cite{Trev01}
has shown that there is no approximation algorithm for this problem with
approximation factor smaller than $\ln B- O(\ln\ln B)$ unless
$\Pcl=\NP$. This leads us to the following theorem.

\begin{theorem}\label{th:blaha-optimal}
  The approximation factor of $\,\ln\ln n$ in Blaha's approximation
  algorithm for minimum base cannot be improved, even for elementary
  abelian $2$-groups, unless $\Pcl=\NP$.
\end{theorem}

\begin{proof}  
  The reduction from $(\log n)$-$\setcover$ to the minimum base problem that is
  explained in the proof of Theorem~\ref{base-mini-hard} preserves the optimal
  solution size. Furthermore, it is easy to see that this reduction carries over
  to all $(\log n)$-$\setcover$ instances. Combined with Trevisan's result, this
  completes the proof.
\end{proof}

\section{The number of non-fixed vertices as parameter}\label{sec:unfixed}

In this section we show that the problems $\nkfixing$ and $\nkbase$
are in~\FPT with running time~$k^{O(k^2)}n^{O(1)}$. We will show this
first for $\nkbase$. We need some permutation group theory.

Let $G\le \Sym(\Omega)$ be a permutation group acting on a
set~$\Omega$. The \emph{support} of a permutation $g\in G$ is
$\supp(g)=\set{i\in\Omega}{i^g\neq i}$. The \emph{orbit} of a point
$i\in\Omega$ is the set $i^G=\set{i^g}{g\in G}$. The group~$G$ is
\emph{transitive} if it has a single orbit in~$\Omega$. Let $G\le
\Sym(\Omega)$ be transitive. A subset $\Delta\subseteq\Omega$ is a
\emph{block} if for every $g\in G$ its
image~$\Delta^g=\set{i^g}{i\in\Delta}$ is either $\Delta^g=\Delta$ or
$\Delta^g\cap\Delta=\emptyset$. Clearly, $\Omega$~and singleton sets
are blocks for any~$G$. All other blocks are called nontrivial. A
transitive group~$G$ is \emph{primitive} if it has no nontrivial
blocks.

There are polynomial-time algorithms that take as input a generating
set for some $G\le \Sym(\Omega)$ and compute its orbits and maximal nontrivial blocks~\cite{Luks}.
We can test if $G$~is primitive in polynomial time.  If
$G$~is transitive on~$\Omega$ we can compute a maximal nontrivial
block~$\Delta_1$. It is easy to see that $\Delta_1^g$~is also a block
for each $g\in G$. This yields a partition of~$\Omega$ into blocks
(which are said to constitute a block system for~$G$):
$\Omega=\Delta_1\sqcup \Delta_2\sqcup\ldots\sqcup \Delta_\ell$.  The
group~$G$ acts transitively on the blocks
$\{\Delta_1,\Delta_2,\ldots,\Delta_\ell\}$. Furthermore, since these
are maximal blocks, the group action is primitive. The following
classic result is useful for our algorithm.

  \begin{lemma}\text{\textup{\cite[Lemma 3.3D]{DixonMortimer}}}\label{jordan}
    Suppose $G\le S_n$ is primitive and~$G$ is neither~$A_n$ nor~$S_n$
    itself. If there is an element $g\in G$ such that
    $\card{\supp(g)}\le k$, then $\card{\Omega}\le (k-1)^{2k}$.
  \end{lemma}
Here, $A_n=\Alt([n])$ denotes the subgroup of~$S_n$ that consists of
those permutations that can be written as the product of an even
number of transpositions.

\begin{theorem}\label{basenk}
There is a $k^{O(k^2)}+k\,n^{O(1)}$ time algorithm for the $\nkbase$ 
problem.
\end{theorem}

\begin{proof}
  Let $G\le S_n$ be the input group given by a generating set and let~$k$
  be the parameter. We call a set $S\subseteq [n]$ a
  \emph{co-base} for~$G$, if $[n]\setminus S$ is a base for~$G$.  The
  algorithm finds a co-base~$S$ of size~$k$ if it exists. During its
  execution, the algorithm may decide to fix some points. Since in
  this case the actual group~$G$ is replaced by the pointwise
  stabilizer subgroup, there is no need to store these points.
The algorithm proceeds as follows.

\begin{enumerate}
\item Let $O_1,O_2,\ldots,O_\ell$ be the orbits of the group~$G$. If
  $\ell\ge k$ then the set~$S$ obtained by picking one point from each
  of the orbits $O_1,O_2,\ldots,O_k$ is a co-base for~$G$.

\item Suppose $\ell<k$, and there is an orbit~$O_i$ of size more
  than~$k^{2k}$ on which $G$'s~action is not primitive.  In this case
  compute a maximal block system of~$G$ in~$O_i$,
  $O_i=\Delta_{i1}\sqcup\ldots\sqcup \Delta_{ir_i}$, and deal with the
  following subcases:
\begin{enumerate}
\item If $r_i>k$, then the set~$S$ obtained by picking one point from
  each block $\Delta_{i1},\ldots,\Delta_{ik}$ is a co-base for~$G$.
\item If $r_i\le k$, then each block~$\Delta_{ij}$ is of size at
  least~$k^{2k-1}$ which is strictly more than~$k$. Thus any $n-k$ sized subset
  of~$[n]$ intersects each block~$\Delta_{ij}$ and hence the support of
  any permutation that moves the blocks. Let~$H$ be the subgroup of~$G$
  that setwise stabilizes all blocks~$\Delta_{ij}$. The subgroup~$H$
  can be computed from~$G$ in polynomial time using the
  Schreier-Sims algorithm~\cite{Luks}. Replace~$G$ by~$H$ and go to
  Step~1. This step is invoked at most $k$~times since each invocation
  increases the number of orbits.
\end{enumerate}

\item Suppose $\ell<k$, and there is an orbit~$O_i$ of size more
  than~$k^{2k}$ such that~$G$ is primitive on~$O_i$, but different from
  $\Sym(O_i)$~and~$\Alt(O_i)$. Then any $k$~points of~$O_i$ form a
  co-base for~$G$ (by Lemma~\ref{jordan}).

\item Suppose there is an orbit~$O_i$ of size more than~$k^{2k}$ such
  that $G$~restricted to~$O_i$ is either~$\Sym(O_i)$ or~$\Alt(O_i)$.
  Then fix the first $\card{O_i}-k$ elements of~$O_i$ (the
  choice of the subset of points fixed does not matter as both
  $\Sym(O_i)$~and~$\Alt(O_i)$ are $t$-transitive for $t\le
  \card{O_i}-2$). Replace~$G$ by the subgroup~$H$ that fixes the first
  $\card{O_i}-k$ elements of~$O_i$ and go to Step~1. This step is 
  invoked at most once.

\item This step is only reached if all orbits are of size at most~$k^{2k}$,
  implying that the entire domain size is at most~$k^{2k+1}$.
  Hence, the algorithm can find a co-base~$S$ of size~$k$
  by brute-force search in $k^{O(k^2)}$~time if it exists.
\end{enumerate}
The brute-force computation (done in the last step), when the search
space is bounded by~$k^{2k+1}$, costs~$k^{O(k^2)}$. The rest of the
computation uses the standard group-theoretic algorithms~\cite{Luks}
whose running time is polynomially bounded by~$n$. Therefore, the
overall running time of the algorithm is bounded 
by~$k^{O(k^2)}+k\,n^{O(1)}$.

We note that the algorithm is in fact a kernelization algorithm.  It
computes in $n^{O(1)}$~time a kernel of size~$k^{2k+1}$ (where size
refers to the size of the domain on which the group acts).
\end{proof}  
We now show the main result of this section, i.e., that $\nkfixing$
is in \FPT.

\begin{theorem}\label{fixingnk}
There is a $k^{O(k^2)}n^{O(1)}$ time algorithm for the $\nkfixing$ problem.
\end{theorem}

\begin{proof}
  Let $X=(V,E,c)$ be a colored $n$-vertex graph and $k$~as parameter be an
  instance of $\nkfixing$. If we can use a subroutine for the Graph
  Isomorphism problem then we can compute a generating set for the
  automorphism group~$\Aut(X)$ of~$X$ with polynomially many calls to
  this subroutine~\cite{Mathon}. With this generating set as input we
  can then run the algorithm of Theorem~\ref{basenk} to compute an
  $(n-k)$ size fixing set for~$X$, if it exists, in time~$k^{O(k^2)}n^{O(1)}$.

  However, it turns out that we can avoid using the Graph Isomorphism
  subroutine and still solve the problem in $k^{O(k^2)}n^{O(1)}$~time
  with the following observations:

\begin{enumerate}
\item We note that any set of size $n-k$ will intersect the support of
  any element $\sigma\in \Aut(X)$ if $\card{\supp(\sigma)}>k$. Thus, we only
  need to collect all elements of support bounded by~$k$.
\item An automorphism $\sigma\in \Aut(X)$ is defined to be a
  \emph{minimal support} automorphism of~$X$ if there is no nontrivial
  automorphism $\varphi\in \Aut(X)$ such that $\supp(\varphi)\subsetneq\supp(\sigma)$.
  For any nontrivial automorphism $\pi\in \Aut(X)$ such that
  $\card{\supp(\pi)}\le k$, there is a minimal support automorphism $\varphi\in
  \Aut(X)$ such that $\card{\supp(\varphi)}\le k$ and $\supp(\varphi)\subseteq
  \supp(\pi)$.
 \item We observe that Schweitzer's algorithm in~\cite{Schweitzer}
  can be used to compute, in $k^{O(k)}n^{O(1)}$~time, the set~$M$ of all minimal
  support automorphisms $\sigma\in \Aut(X)$ such that
  $\card{\supp(\sigma)}\le k$.
\item Let $G'$ be the subgroup of~$\Aut(X)$ generated by~$M$. It follows from
  the above discussion that an $n-k$ sized subset of~$V$ is a base for~$\Aut(X)$
  (and thus a fixing set for~$X$) if and only if it is a base for~$G'$. We can
  apply the algorithm of Theorem~\ref{basenk} to compute such a base if it
  exists.
  \qedhere
\end{enumerate}
\end{proof}

\section{The number of individualized vertices as parameter}\label{sec:indiv}

In this section, we show that the problem \IndC is \WP-hard for all
classes~$\mathcal{C}$ of the color refinement
hierarchy~\eqref{eq:cr-hier}. To this end, we give a reduction from
\WMCS, which is known to be \WP-complete~\cite{ADF95}.

\begin{problem}{\WMCS}
 \item[Input:] A monotone boolean circuit~$C$ on~$n$ inputs and an integer~$k$
 \item[Parameter:] $k$
 \item[Question:] Is there an assignment $x\in\{0,1\}^n$ of Hamming weight~$k$
  so that $C(x)=1$?
\end{problem}

\begin{theorem}\label{th:wp-hard}
  For all classes~$\mathcal{C}$ of the color refinement
  hierarchy~\eqref{eq:cr-hier}, \IndC is $\WP$-hard, even for graphs of color
  class size at most~$4$.
\end{theorem}

\begin{proof}
  We will give a parameter-preserving reduction that maps positive
  instances of \WMCS to positive instances of \IndDiscrete, while
  negative instances are mapped to negative instances of
  \IndRefinable. A similar reduction was used to show that the classes
  from the color refinement hierarchy~\eqref{eq:cr-hier} are all
  \Pcl-hard~\cite{AKRV}, which in turn builds on ideas of
  Grohe~\cite{Gro99}.
  
  Let $\pair{C,k}$ be the given instance of \WMCS, and let~$n$ be the
  number of inputs of the circuit~$C$. We define a graph~$X_C$. For each
  gate~$g_k$ of~$C$ (including the input gates), $X_C$~contains a
  vertex pair $P_k=\{v_k,v'_k\}$, which forms a color class. If a pair
  corresponds to an input gate, we call it an \emph{input pair.} The
  intention is that setting an input~$g_i$ to~$1$ corresponds to
  individualizing the vertex~$v_i$; we will add gadgets to~$X_C$ so
  that after color refinement it holds also for each non-input
  gate~$g_k$ that $g_k=1$ if and only if $v_k$~and~$v'_k$ have
  different colors.

  To achieve this, we use the gadgets given in
  Figure~\ref{fig:gadgets}. The basic building block is the gadget
  $\CFI(P_i,P_j,P_k)$ introduced by Cai, Fürer, and
  Immerman~\cite{CFI92}. It connects the three pairs
  $P_i$,~$P_j$,~and~$P_k$ using four additional vertices as depicted.
  These four vertices form a color class~$F$; each instance of the
  gadget uses its own copy of~$F$. This gadget has the property that
  every automorphism flips either none or exactly two of the pairs
  $P_i$,~$P_j$~and~$P_k$; thus the \CFI-gadget implements the
  \textsc{xor}~function in the sense that any automorphism must
  flip~$P_k$ if and only if it flips exactly one of $P_i$~and~$P_j$.
  In our case, however, the \CFI-gadget implements the
  \textsc{and}~function: If both $P_i$~and~$P_j$ are distinguished
  (either by direct individualization or in previous rounds of color
  refinement), the vertices of the inner color class~$F$ and
  consequently~$P_k$ will be distinguished in two rounds of color
  refinement. Conversely, if at most one of the pairs $P_i$~and~$P_j$
  is distinguished, then the color class~$F$ is split into two color
  classes of size~$2$ and color refinement stops at this point,
  leaving the other two pairs non-distinguished. For each
  \textsc{and}~gate $g_k=g_i\land g_j$ in~$C$, we add the gadget
  $\CFI(P_i,P_j,P_k)$ to~$X_C$.

  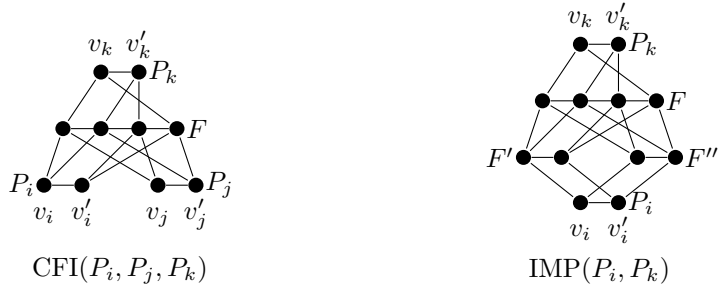
\begin{figure}
    \centering
    \begin{tikzpicture}[y=1.5cm]
      \tikzset{gv/.style={circle,fill,inner sep=2pt},
        every edge/.style={draw,shorten <=.5pt,shorten >=.5pt},
        every label/.style={text height=height("$v'_k$"),text depth=depth("$v'_k$"),inner sep=2pt},
      }
      \path
      (0,-.75) node {$\CFI(P_i,P_j,P_k)$}
      (-1,0) node[gv,label=below:$v_i$] (vi) {}
      (-.5,0) node[gv,label=below:$v'_i$] (vpi) {}
      (.5,0) node[gv,label=below:$v_j$] (vj) {}
      (1,0) node[gv,label=below:$v'_j$] (vpj) {}
      (-3/4,.5) node[gv] (f00) {} edge (vi) edge (vj)
      (-1/4,.5) node[gv] (f01) {} edge (vi) edge (vpj)
      (1/4,.5) node[gv] (f10) {} edge (vpi) edge (vj)
      (3/4,.5) node[gv] (f11) {} edge (vpi) edge (vpj)
      (-1/4,1) node[gv,label=above:$v_k$] (vk) {} edge (f00) edge (f11)
      (1/4,1) node[gv,label=above:$v'_k$] (vpk) {} edge (f01) edge (f10);
      \begin{scope}[on background layer={gray!50!white,
          line cap=round,line width=10pt,shorten <=-3pt,shorten >=-3pt,
        every node/.style={black,text depth=0pt,inner sep=2pt}}]
        \draw (vi) -- (vpi) node[at start,left=3pt] {$P_i$};
        \draw (vj) -- (vpj) node[right] {$P_j$};
        \draw (vk) -- (vpk) node[right] {$P_k$};
        \draw (f00) -- (f11) node[right] {$F$};
      \end{scope}
    \end{tikzpicture}\hfil
    \begin{tikzpicture}[y=1.5cm]
      \tikzset{gv/.style={circle,fill,inner sep=2pt},
        every edge/.style={draw,shorten <=.5pt,shorten >=.5pt},
        every label/.style={text height=height("$v'_k$"),text depth=depth("$v'_k$"),inner sep=2pt},
      }
      \path
      (0,-.5) node {$\IMP(P_i,P_k)$}
      (-.25,.1) node[gv,label=below:$v_i$] (vi) {}
      (.25,.1) node[gv,label=below:$v'_i$] (vpi) {}
      (-1,.5) node[gv] (wi) {} edge (vi)
      (-.5,.5) node[gv] (wpi) {} edge (vpi)
      (.5,.5) node[gv] (vj) {} edge (vi)
      (1,.5) node[gv] (vpj) {} edge (vpi)
      (-3/4,1) node[gv] (f00) {} edge (wi) edge (vj)
      (-1/4,1) node[gv] (f01) {} edge (wi) edge (vpj)
      (1/4,1) node[gv] (f10) {} edge (wpi) edge (vj)
      (3/4,1) node[gv] (f11) {} edge (wpi) edge (vpj)
      (-1/4,1.5) node[gv,label=above:$v_k$] (vk) {} edge (f00) edge (f11)
      (1/4,1.5) node[gv,label=above:$v'_k$] (vpk) {} edge (f01) edge (f10);
      \begin{scope}[on background layer={gray!50!white,
          line cap=round,line width=10pt,shorten <=-3pt,shorten >=-3pt,
        every node/.style={black,text depth=0pt,inner sep=2pt}}]
        \draw (vi) -- (vpi) node[right] {$P_i$};
        \draw (wi) -- (wpi) node[at start,left=3pt] {$F'$};
        \draw (vj) -- (vpj) node[right] {$F''$};
        \draw (f00) -- (f11) node[right] {$F$};
        \draw (vk) -- (vpk) node[right] {$P_k$};
      \end{scope}
    \end{tikzpicture}
    \caption{Gadgets used in the reduction of Theorem~\ref{th:wp-hard}}
    \label{fig:gadgets}
  \end{figure}

  The second gadget we use is $\IMP(P_i,P_k)$. It consists of the gadget
  $\CFI(F',F'',P_k)$, where $F'$~and~$F''$ are vertex pairs that form color
  classes of size two, and perfect matchings that connect these pairs to~$P_i$;
  see Fig.~\ref{fig:gadgets}. Again, each instance of this gadget gets its own
  copy of the color classes $F$,~$F'$~and~$F''$. There is an automorphism
  of~$\IMP(P_i,P_k)$ that flips the vertices in~$P_i$, but none that flips the
  vertices in~$P_k$. In the color refinement setting, this gadget implements the
  \textsc{implication}~function: When~$P_i$ is distinguished, this will
  propagate to both~$F'$~and~$F''$, and consequently also to~$F$ and~$P_k$.
  Conversely, distinguishing~$P_k$ will only split~$F$ into two color classes of
  size~$2$ before color refinement stops. For each \textsc{or}~gate $g_k=g_i\lor
  g_j$ in~$C$, we add the gadgets $\IMP(P_i,P_k)$ and $\IMP(P_j,P_k)$ to~$X_C$.
  For the output gate~$g_\ell$ of~$C$, we add a second vertex pair~$Q$ and the
  gadget~$\IMP(P_\ell,Q)$ to~$X_C$.

  Our above analysis of the gadgets ensures that the following
  invariant holds when running color refinement on~$X_C$ after
  individualizing a subset of its input pairs: For each
  \textsc{implication} gadget $\IMP(P_i,P_k)$ in~$X_C$ the pair~$P_k$
  can only be distinguished if $P_i$~is distinguished, and for each
  \textsc{and} gadget $\CFI(P_i,P_j,P_k)$ the pair~$P_k$ can only be
  distinguished if both $P_i$~and~$P_j$ are distinguished. This
  implies the following.
  \begin{claim}
    Running color refinement on~$X_C$ after individualizing some input
    pairs will distinguish exactly those pairs~$P_k$ for which the
    gate~$g_k$ evaluates to~$1$ under the assignment that sets exactly
    those input gates to~$1$ whose corresponding pairs were initially
    individualized.
  \end{claim}

  Let~$X'_C$ be the graph that is obtained from~$X_C$ by adding
  implication gadgets from~$Q$ to each pair~$P_i$ that corresponds to
  an input gate~$g_i$. If $C$~has a satisfying
  assignment~$x\in\{0,1\}^n$ of weight~$k$, individualizing the
  vertices~$v_i$ with $x_i=1$ and subsequently running color
  refinement will assign distinct colors to all vertices of~$X_C$.
  Indeed, the gadgets of~$X_C$ ensure that the pair~$Q$ becomes
  distinguished, the additional gadgets in~$X'_C$ propagate this to
  all input pairs~$P_i$, and the gates of~$X_C$ in turn make sure that
  all remaining color classes become distinguished. Conversely, if
  $C$~does not have a weight~$k$ satisfying assignment, there is no
  way to individualize $k$~input vertices such that color refinement
  distinguishes~$Q$. However, we already noted that there is no
  automorphism that transposes the output pair of the $\IMP(P_\ell,Q)$
  gadget, so no way of individualizing $k$~input vertices makes~$X'_C$
  refinable.

  In~$X'_C$, it always suffices to individualize one vertex from~$Q$
  to make it discrete. To drop the assumption that each of the
  $k$~individualized vertices must correspond to an input gate, we
  construct a graph~$X''_C$. It consists of $n$~input
  pairs~$P_i=\{v_i,v'_i\}$ and $n$~copies of~$X_C$, to which we will
  refer to as $X_C^{(1)},\dotsc,X_C^{(n)}$. We also add the gadgets
  $\IMP(P_i,P_i^{(j)})$ for all $i,j\in\{1,\dotsc,n\}$ and the gadgets
  $\IMP(Q^{(i)},P_i)$ for all $i\in\{1,\dotsc,n\}$.
    We will show that $\pair{C,k}\mapsto\pair{X''_C,k}$ is the desired
    reduction.

    Individualizing $k$~input vertices that correspond to a satisfying
    assignment makes~$X''_C$ discrete, this happens for the same
    reason as in~$X'_C$. Conversely, let~$U$ be a set of $k$~vertices
    so that individualizing them makes~$X''_C$ refinable. Let
    \[
    I =  \set*{i\in[n]}{U\text{ contains a vertex of }P_i \text{ or }X_C^{(i)}\text{, or an inner vertex}\\
      \text{of }\IMP(Q^{(i)},P_i) \text{ or of }\IMP(P_i,P_i^{(j)})\text{ for some }j}.
    \]
    The only way individualizing~$U$ and subsequent color refinement
    can affect a copy~$X_C^{(j)}$ of~$X_C$ with
    $j\in\{1,\dotsc,n\}\setminus I$ is via the pairs~$P_i$, $i\in I$.
    Indeed, the gadget $\IMP(Q^{(j)},P_j)$ cannot cause~$Q^{(j)}$ to
    be distinguished, and if for some $j'\in\{1,\dotsc,n\}\setminus I$
    the pair~$P_{j'}$ becomes distinguished, then whatever color
    refinement did in~$X_C^{(j')}$ will also apply to~$X_C^{(j)}$
    before~$P_{j'}$ becomes distinguished. In particular, after
    individualizing $U'=\{v_i\mid i\in I\}$ instead of~$U$, color
    refinement must distinguish the pair~$Q^{(j)}$; otherwise this
    pair would be a color class of the stable coloring of~$X''_C$
    after individualizing~$U$, contradicting its refinability. Thus
    setting the inputs given by~$I$ to~$1$ must satisfy~$C$. As
    $\card{I}\le\card{U}=k$ and $C$~is monotone, this implies that
    $C$~has a satisfying assignment of weight~$k$.
\end{proof}
As a corollary to this proof we can derive the
\WP-hardness of the $\kcoldeg$ problem.
 
\begin{corollary}\label{cor:coldeg}
  $\kcoldeg$ is $\WP$-hard.
\end{corollary}

\begin{proof}
  In the previous reduction we mapped instances of \WMCS to instances
  of $\kdiscrete$ such that the given boolean circuit~$C$ has a
  satisfying assignment of weight~$k$ if and only if the resulting
  graph~$X''_C$ can be made discrete by individualizing~$k$ vertices.
  Note that individualizing $k$~vertices in~$X''_C$ and subsequently
  running color refinement results in singleton color classes if and
  only if it brings the color valence down to~$0$. Thus, $\kcoldeg$ is
  $\WP$-hard even for $d=0$.
\end{proof}

\subsection{Graphs of color class size at most 3}

We call a vertex-colored graph \emph{$b$-bounded} if all its color
classes are of size at most~$b$. In this section, we show that for any
$3$-bounded graph, we can compute in polynomial time the minimum
number of vertices that have to be individualized so that the
resulting colored graph becomes rigid, discrete, amenable, compact, or
refinable.
  We end this section by providing sufficient conditions for a 3-bounded
  graph to be compact. We first recall the definition of compactness.
  Let~$A$ be the adjacency matrix of a graph $X$. A doubly stochastic
  matrix~$Y$ is said to be a \emph{fractional automorphism} of~$X$ if it
  satisfies the system of linear equations $AY=Y\!\!A$. A graph~$X$ is called
  \emph{compact} if every fractional automorphism of~$X$ can be
  expressed as a convex combination of some permutation matrices
  corresponding to automorphisms of~$X$.  For a graph with color classes
  $C_1,\dots,C_m$, a fractional automorphism is a block diagonal matrix
  with submatrices $Y_1,\dots,Y_m$.  Here, the matrix~$Y_i$ is a
  $\card{C_i}\times\card{C_i}$ matrix.

  \begin{lemma}\label{lem:compact1}
    Let~$X$ be a $3$-bounded graph whose color classes are stable. If
    $\Aut(X)$~restricted to any color class~$C_i$ of~$X$ is the full
    symmetric group on~$C_i$, then $X$~is compact.
  \end{lemma}

  \begin{proof}
    As argued in the proof of Theorem~\ref{cc3}, between any two color
    classes we either have a perfect matching or no edges at
    all. Further, we can assume that the color classes of~$X$ are all
    linked to each other. (Otherwise we can partition the vertex set
    $V=V_1\sqcup\cdots\sqcup V_l$ such that each set~$V_i$ is a union of
    linked color classes and there are no edges between $V_i$~and~$V_j$
    whenever $i\ne j$, implying that $X$~is compact if each of the
    induced subgraphs~$X[V_i]$ is compact.)

    Since $\Aut(X)$~restricted to any color class~$C_i$ of~$X$ is the
    full group on~$C_i$ and the color classes of~$X$ are all linked to
    each other, it follows that $X$~has exactly $b$~components.  We can
    number these components, and hence the vertices inside any color
    class, from~$1$ to~$b$.
    
    \begin{claim}\label{comp-mat}
      Let~$Y$ be a fractional automorphism of~$X$.  If a matching
      between color classes $C_i$~and~$C_j$ connects vertices
      $x,y \in C_i$ with $x',y'\in C_j$ respectively, then
      $Y_{x,y}=Y_{x',y'}$.
    \end{claim}
    Expanding the system of linear equations $AY=Y\!\!A$, we obtain the
    subsystem $A_{ij} Y_{j} = Y_{i} A_{ij}$ where $A_{ij}$~is the
    adjacency matrix of~$X[C_i,C_j]$ and $Y_i,Y_j$ are the fractional
    automorphisms induced on $C_i$~and~$C_j$, respectively. Further
    expanding this subsystem proves the claim.

    We now finish the proof of the lemma. Let~$Y_i$ be the $b\times b$ submatrix of the fractional
    automorphism~$Y$ restricted to color class~$C_i$. By the above claim,
    the $(i,j)^{th}$ entry of the submatrices $Y_1,\dots,Y_m$ must be equal. Therefore, 
    $Y_1 = Y_2 = \dots = Y_m = Y^*$ for some doubly stochastic $b \times b$ matrix~$Y^*$. 
    By Birkhoff's theorem (see, e.g.~\cite{Bru}), we can write~$Y^*$ as a convex combination of $b!$~permutation
    matrices $P_1,\dots,P_{b!}$. Since $Y$~is a block diagonal matrix with $m$~blocks of~$Y^*$,
    we can similarly rewrite~$Y$ as a convex combination of $b!$~permutation matrices 
    $\hat{P_1},\dots,\hat{P_b}$. Here, $\hat{P_i}$~is block diagonal with $m$~blocks of~$P_i$. 
    Since $X$~has exactly~$b$ connected components, 
    it is easy to see that $\hat{P_1},\dots,\hat{P_b}$ are automorphisms of~$X$. 
    Hence, the graph~$X$ is compact.
  \end{proof}

  \begin{lemma}\label{lem:compact}
    Let~$X$ be a connected $3$-bounded graph whose color classes are
    stable. If some $\sigma\in\Aut(X)$ is cyclic (i.e., $\sigma$~acts
    cyclically on each color class~$C_i$), then $X$~is compact.
  \end{lemma}

  \begin{proof}
    We first prove two claims.
    \begin{claim}\label{comp-link}
      Let~$\sigma$ be an automorphism of~$X$. If there is a path
      between two vertices $u$~and~$v$, then for any fractional
      automorphism~$Y$ of~$X$ it holds that
      $Y_{u,\sigma(u)} = Y_{v,\sigma(v)}$.
    \end{claim}
    If vertices $u$~and~$v$ are connected by a path
    $u$-$u_1$-$\dots$-$u_l$-$v$ of matching edges, the vertices
    $\sigma(u)$~and~$\sigma(v)$ are also connected by a parallel path
    $\sigma(u)$-$\sigma(u_1)$-$\dots$-$\sigma(u_l)$-$\sigma(v)$ of
    matching edges.  Applying Claim~\ref{comp-mat} repeatedly along the
    above two matching paths proves the claim.

    \begin{claim}\label{comp-dec}
      Let the color class~$C_i$ be the set of vertices $\{u_i,v_i,w_i\}$.
      Suppose the cyclic automorphism~$\sigma$ sends $u_i,v_i,w_i$ to
      $v_i,w_i,u_i$ respectively.  If $Y$~is a fractional automorphism
      of~$X$, there exist $\alpha,\beta,\gamma \in [0,1]$ such that
      $\alpha+\beta+\gamma=1$ and
      \begin{align*}
        Y_{u_i,u_i} = Y_{v_i,v_i} = Y_{w_i,w_i} = \alpha \quad \mbox{for all $i\in [n]$} \\
        Y_{u_i,v_i} = Y_{v_i,w_i} = Y_{w_i,v_i} = \beta \quad \mbox{for all $i\in [n]$} \\
        Y_{u_i,w_i} = Y_{v_i,u_i} = Y_{w_i,v_i} = \gamma \quad \mbox{for all $i\in [n]$}
      \end{align*}
    \end{claim}
    To prove the claim it suffices to observe that between every two
    vertices there is a path in~$X$. Hence, we can apply Claim~\ref{comp-link}
    for the three cyclic automorphisms
    $\{id,\sigma,\sigma^2\}$ to obtain the three equations respectively.

    Now we are ready to show that $X$~is compact.  Using Claim~\ref{comp-dec},
    a fractional automorphism~$Y$ of~$X$ can be decomposed
    as a convex combination $\alpha I_1 + \beta I_2 + \gamma I_3$ where
    $I_1,I_2,I_3$ are the permutation matrices corresponding to the three
    cyclic automorphisms.
  \end{proof}

\begin{theorem}\label{cc3}
  For any 3-bounded graph we can compute in polynomial time a vertex
  set~$S$ of minimum size such that individualizing (or fixing) all
  the vertices in~$S$ makes the graph discrete, amenable, compact,
  refinable (or rigid).
\end{theorem}

\begin{proof}
  Let $X=(V,E,c)$ be the given 3-bounded graph. We first compute the
  color partition $\{C_1,\ldots,C_m\}$ of the stable coloring of~$X$.
  We can assume that each induced graph $X_i=X[C_i]$ is empty and
  each induced bipartite graph $X_{ij}=X[C_i,C_j]$ has at most
  $\card{C_i}\cdot\card{C_j}/2$ edges, as otherwise we can complement these
  subgraphs. Since the partition $\{C_1,\ldots,C_m\}$ is stable
  and the color classes have size at most~$3$, it
  follows that there are no edges between color classes having
  different sizes, and that between color classes $C_i$~and~$C_j$ of
  the same size we either have a perfect matching or no edges at
  all. 

  We say that two color classes $C_i$~and~$C_j$ are \emph{linked} if
  there is a path between some vertex $u\in C_i$ and some vertex
  $v\in C_j$. Since this is an equivalence relation, it partitions the
  color classes into equivalence classes. This induces a partition
  $V=V_1\sqcup\cdots\sqcup V_l$ of the vertices such that each set~$V_i$
  is a union of linked color classes having the same size and
  there are no edges between $V_i$~and~$V_j$ whenever $i\ne j$.
  Hence, it suffices to solve the problem separately for each of the
  induced subgraphs~$X[V_i]$.

  If all color classes of~$X[V_i]$ are of size~2, then $\Aut(X[V_i])$~contains
  exactly one non-trivial automorphism flipping all the color
  classes, implying that $X[V_i]$~is compact (see
  Lemma~\ref{lem:compact1}). In this case it suffices to individualize
  (or fix) an arbitrary vertex to make the graph discrete (or
  rigid). Further, $X[V_i]$~is already amenable if and only if it is a
  forest~\cite{AKRV2}.

  If all color classes of~$X[V_i]$ are of size~3, then we compute its
  connected components as well as~$\Aut(X[V_i])$ (which is even
  possible in logspace~\cite{JKMT03,Rei05}) and consider the following
  subcases.
  \begin{itemize}
  \item If $X[V_i]$ has 6~automorphisms (or, equivalently, consists of
    three connected components), then $X[V_i]$~is compact (see
    Lemma~\ref{lem:compact1}) and it suffices to individualize two
    vertices inside an arbitrary color class to make the graph
    discrete. On the other hand, if we individualize only one vertex,
    then the graph does not become discrete (not even rigid). Further,
    $X[V_i]$~is amenable if and only if it is a forest~\cite{AKRV2}. If
    $X[V_i]$~contains cycles then we need to individualize 2~vertices to make
    the graph amenable.
  \item If $X[V_i]$ has 3~automorphisms, then it follows that these
    automorphisms act cyclically on each color class and $X[V_i]$~is
    connected as well as compact (see Lemma~\ref{lem:compact}). In
    this case it suffices to individualize an arbitrary vertex to make
    the graph discrete.
  \item If $X[V_i]$ has 2~automorphisms (or, equivalently, consists of
    two connected components), then $X[V_i]$~is not refinable and it
    suffices to individualize an arbitrary vertex in the larger of the
    two components to make the graph discrete.
  \item Finally, if $X[V_i]$~is rigid, then it follows that $X[V_i]$~is
    connected and not refinable. In this case it suffices to
    individualize an arbitrary vertex to make the graph discrete.
    \qedhere
  \end{itemize}
\end{proof}
  We next show that for any 3-bounded graph the stable color partition
  is computable in logspace. Combined with the case analysis in the
  proof of Theorem~\ref{cc3} it follows that for any 3-bounded graph,
  the minimum number of vertices that have to be individualized (or
  fixed) so that the resulting colored graph becomes discrete,
  amenable, compact, refinable (or rigid) is even computable in
  logspace.

  \begin{lemma}\label{lem:stable}
    The stable color partition of any 3-bounded graph is computable in
    logspace.
  \end{lemma}

  \begin{proof}
    Let $X=(V,E,c)$ be a $3$-bounded graph and let $C_1,\ldots,C_m$ be its color
    classes. We use~$X_i$ to denote the graph~$X[C_i]$ induced
    by~$C_i$ and $X_{ij}$~to denote the bipartite graph~$X[C_i,C_j]$
    induced by the pair of color classes $C_i$~and~$C_j$.  We can assume
    that the vertices in each graph~$X_i$ have the same
    degree. Otherwise we can split~$C_i$ into smaller color
    classes. Moreover, we can assume that each graph~$X_i$ is the empty
    graph on vertex set~$C_i$ and that each bipartite graph~$X_{ij}$ has
    at most $\card{C_i}\cdot\card{C_j}/2$ edges, since otherwise, we can
    replace~$X_{ij}$ by the complement bipartite graph.

    The idea is to pick a set~$W$ of vertices and a set $F\subseteq E$
    of edges such that color refinement assigns a unique color to all
    vertices that are reachable from some vertex in~$W$ via edges in~$F$.
    A vertex belongs to~$W$ if it receives a unique color after the
    first round (vertices belonging to~$W$ are depicted as a box in
    Fig.~\ref{fig:connections}). The edge set~$F$ is formed by picking
    from each graph~$X_{ij}$ all edges $e=\{v,w\}$ with
    $e\cap W=\emptyset$ such that individualizing one of the two
    endpoints of~$e$ causes color refinement to assign a unique color
    also to the other endpoint (see Fig.~\ref{fig:connections}; these
    edges are depicted in bold).  It is clear that $W$~and~$F$ can be
    easily determined in logspace.

    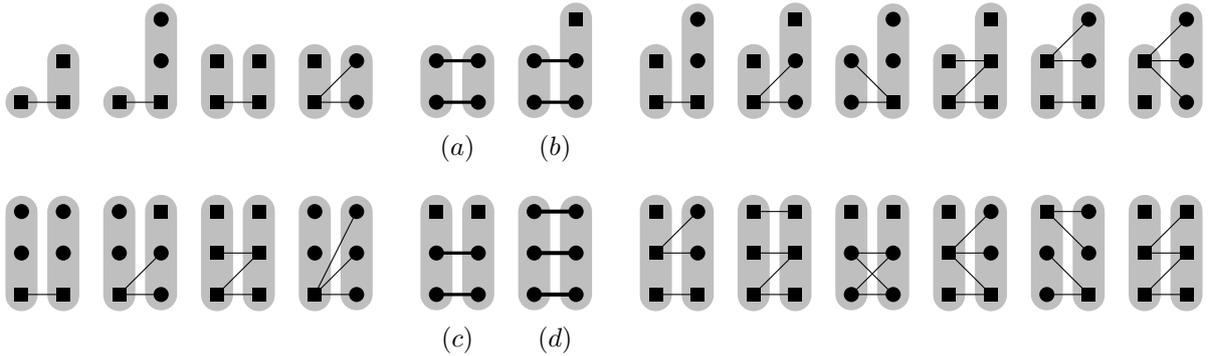
\begin{figure}
      \centering
      \tikzset{
        x=.275cm,y=.55cm,
        gv/.style={circle,fill,inner sep=2pt},
        gu/.style={fill,inner sep=2.5pt},
        every edge/.style={draw,shorten <=-1pt,shorten >=-1pt},
        cc/.style={gray!50!white,line cap=round,line width=12pt,shorten <=-3pt,shorten >=-3pt},
      }
      \begin{tikzpicture}
        \path
        (0,-1.1) node {\phantom{$(d)$}}
        (-1,0) node[gu] (v3) {}
        (1,1) node[gu] (u2) {}
        (1,0) node[gu] (u3) {} edge (v3);
        \begin{scope}[on background layer]
          \fill[cc] (v3) circle (6pt);
          \draw[cc] (u2) -- (u3);
        \end{scope}
      \end{tikzpicture}\hfill
      \begin{tikzpicture}
        \path
        (0,-1.1) node {\phantom{$(d)$}}
        (-1,0) node[gu] (v3) {}
        (1,2) node[gv] (u1) {}
        (1,1) node[gv] (u2) {}
        (1,0) node[gu] (u3) {} edge (v3);
        \begin{scope}[on background layer]
          \fill[cc] (v3) circle (6pt);
          \draw[cc] (u1) -- (u3);
        \end{scope}
      \end{tikzpicture}\hfill
      \begin{tikzpicture}
        \path
        (0,-1.1) node {\phantom{$(d)$}}
        (-1,1) node[gu] (v2) {}
        (-1,0) node[gu] (v3) {}
        (1,1) node[gu] (u2) {}
        (1,0) node[gu] (u3) {} edge (v3);
        \begin{scope}[on background layer]
          \draw[cc] (v2) -- (v3);
          \draw[cc] (u2) -- (u3);
        \end{scope}
      \end{tikzpicture}\hfill
      \begin{tikzpicture}
        \path
        (0,-1.1) node {\phantom{$(d)$}}
        (-1,1) node[gu] (v2) {}
        (-1,0) node[gu] (v3) {}
        (1,1) node[gv] (u2) {} edge (v3)
        (1,0) node[gv] (u3) {} edge (v3);
        \begin{scope}[on background layer]
          \draw[cc] (v2) -- (v3);
          \draw[cc] (u2) -- (u3);
        \end{scope}
      \end{tikzpicture}\hfill\hfill
      \begin{tikzpicture}
        \path
        (0,-1.1) node {\phantom{$(d)$}}
        (0,-1.1) node {$(a)$}
        (-1,1) node[gv] (v2) {}
        (-1,0) node[gv] (v3) {}
        (1,1) node[gv] (u2) {} edge[line width=1.2pt] (v2)
        (1,0) node[gv] (u3) {} edge[line width=1.2pt] (v3);
        \begin{scope}[on background layer]
          \draw[cc] (v2) -- (v3);
          \draw[cc] (u2) -- (u3);
        \end{scope}
      \end{tikzpicture}\hfill
      \begin{tikzpicture}
        \path
        (0,-1.1) node {\phantom{$(d)$}}
        (0,-1.1) node {$(b)$}
        (-1,1) node[gv] (v2) {}
        (-1,0) node[gv] (v3) {}
        (1,2) node[gu] (u1) {}
        (1,1) node[gv] (u2) {} edge[line width=1.2pt] (v2)
        (1,0) node[gv] (u3) {} edge[line width=1.2pt] (v3);
        \begin{scope}[on background layer]
          \draw[cc] (v2) -- (v3);
          \draw[cc] (u1) -- (u3);
        \end{scope}
      \end{tikzpicture}\hfill\hfill
      \begin{tikzpicture}
        \path
        (0,-1.1) node {\phantom{$(d)$}}
        (0,-1.1) node {\phantom{$(d)$}}
        (-1,1) node[gu] (v2) {}
        (-1,0) node[gu] (v3) {}
        (1,2) node[gv] (u1) {}
        (1,1) node[gv] (u2) {}
        (1,0) node[gu] (u3) {} edge (v3);
        \begin{scope}[on background layer]
          \draw[cc] (v2) -- (v3);
          \draw[cc] (u1) -- (u3);
        \end{scope}
      \end{tikzpicture}\hfill
      \begin{tikzpicture}
        \path
        (0,-1.1) node {\phantom{$(d)$}}
        (-1,1) node[gu] (v2) {}
        (-1,0) node[gu] (v3) {}
        (1,2) node[gu] (u1) {}
        (1,1) node[gv] (u2) {} edge (v3)
        (1,0) node[gv] (u3) {} edge (v3);
        \begin{scope}[on background layer]
          \draw[cc] (v2) -- (v3);
          \draw[cc] (u1) -- (u3);
        \end{scope}
      \end{tikzpicture}\hfill
      \begin{tikzpicture}
        \path
        (0,-1.1) node {\phantom{$(d)$}}
        (-1,1) node[gv] (v2) {}
        (-1,0) node[gv] (v3) {}
        (1,2) node[gv] (u1) {}
        (1,1) node[gv] (u2) {}
        (1,0) node[gu] (u3) {} edge (v2) edge (v3);
        \begin{scope}[on background layer]
          \draw[cc] (v2) -- (v3);
          \draw[cc] (u1) -- (u3);
        \end{scope}
      \end{tikzpicture}\hfill
      \begin{tikzpicture}
        \path
        (0,-1.1) node {\phantom{$(d)$}}
        (-1,1) node[gu] (v2) {}
        (-1,0) node[gu] (v3) {}
        (1,2) node[gu] (u1) {}
        (1,1) node[gu] (u2) {} edge (v2) edge (v3)
        (1,0) node[gu] (u3) {} edge (v3);
        \begin{scope}[on background layer]
          \draw[cc] (v2) -- (v3);
          \draw[cc] (u1) -- (u3);
        \end{scope}
      \end{tikzpicture}\hfill
      \begin{tikzpicture}
        \path
        (0,-1.1) node {\phantom{$(d)$}}
        (-1,1) node[gu] (v2) {}
        (-1,0) node[gu] (v3) {}
        (1,2) node[gv] (u1) {} edge (v2)
        (1,1) node[gv] (u2) {} edge (v2)
        (1,0) node[gu] (u3) {} edge (v3);
        \begin{scope}[on background layer]
          \draw[cc] (v2) -- (v3);
          \draw[cc] (u1) -- (u3);
        \end{scope}
      \end{tikzpicture}\hfill
      \begin{tikzpicture}
        \path
        (0,-1.1) node {\phantom{$(d)$}}
        (-1,1) node[gu] (v2) {}
        (-1,0) node[gu] (v3) {}
        (1,2) node[gv] (u1) {} edge (v2)
        (1,1) node[gv] (u2) {} edge (v2)
        (1,0) node[gv] (u3) {} edge (v2);
        \begin{scope}[on background layer]
          \draw[cc] (v2) -- (v3);
          \draw[cc] (u1) -- (u3);
        \end{scope}
      \end{tikzpicture}\\[4mm]%
      \begin{tikzpicture}
        \path
        (0,-1.1) node {\phantom{$(d)$}}
        (-1,2) node[gv] (v1) {}
        (-1,1) node[gv] (v2) {}
        (-1,0) node[gu] (v3) {}
        (1,2) node[gv] (u1) {}
        (1,1) node[gv] (u2) {}
        (1,0) node[gu] (u3) {} edge (v3);
        \begin{scope}[on background layer]
          \draw[cc] (v1) -- (v3);
          \draw[cc] (u1) -- (u3);
        \end{scope}
      \end{tikzpicture}\hfill
      \begin{tikzpicture}
        \path
        (0,-1.1) node {\phantom{$(d)$}}
        (-1,2) node[gv] (v1) {}
        (-1,1) node[gv] (v2) {}
        (-1,0) node[gu] (v3) {}
        (1,2) node[gu] (u1) {}
        (1,1) node[gv] (u2) {} edge (v3)
        (1,0) node[gv] (u3) {} edge (v3);
        \begin{scope}[on background layer]
          \draw[cc] (v1) -- (v3);
          \draw[cc] (u1) -- (u3);
        \end{scope}
      \end{tikzpicture}\hfill
      \begin{tikzpicture}
        \path
        (0,-1.1) node {\phantom{$(d)$}}
        (-1,2) node[gu] (v1) {}
        (-1,1) node[gu] (v2) {}
        (-1,0) node[gu] (v3) {}
        (1,2) node[gu] (u1) {}
        (1,1) node[gu] (u2) {} edge (v2) edge (v3)
        (1,0) node[gu] (u3) {} edge (v3);
        \begin{scope}[on background layer]
          \draw[cc] (v1) -- (v3);
          \draw[cc] (u1) -- (u3);
        \end{scope}
      \end{tikzpicture}\hfill
      \begin{tikzpicture}
        \path
        (0,-1.1) node {\phantom{$(d)$}}
        (-1,2) node[gv] (v1) {}
        (-1,1) node[gv] (v2) {}
        (-1,0) node[gu] (v3) {}
        (1,2) node[gv] (u1) {} edge (v3)
        (1,1) node[gv] (u2) {} edge (v3)
        (1,0) node[gv] (u3) {} edge (v3);
        \begin{scope}[on background layer]
          \draw[cc] (v1) -- (v3);
          \draw[cc] (u1) -- (u3);
        \end{scope}
      \end{tikzpicture}\hfill\hfill
      \begin{tikzpicture}
        \path
        (0,-1.1) node {\phantom{$(d)$}}
        (0,-1.1) node {$(c)$}
        (-1,2) node[gu] (v1) {}
        (-1,1) node[gv] (v2) {}
        (-1,0) node[gv] (v3) {}
        (1,2) node[gu] (u1) {}
        (1,1) node[gv] (u2) {} edge[line width=1.2pt] (v2)
        (1,0) node[gv] (u3) {} edge[line width=1.2pt] (v3);
        \begin{scope}[on background layer]
          \draw[cc] (v1) -- (v3);
          \draw[cc] (u1) -- (u3);
        \end{scope}
      \end{tikzpicture}\hfill
      \begin{tikzpicture}
        \path
        (0,-1.1) node {\phantom{$(d)$}}
        (0,-1.1) node {$(d)$}
        (-1,2) node[gv] (v1) {}
        (-1,1) node[gv] (v2) {}
        (-1,0) node[gv] (v3) {}
        (1,2) node[gv] (u1) {} edge[line width=1.5pt] (v1)
        (1,1) node[gv] (u2) {} edge[line width=1.5pt] (v2)
        (1,0) node[gv] (u3) {} edge[line width=1.5pt] (v3);
        \begin{scope}[on background layer]
          \draw[cc] (v1) -- (v3);
          \draw[cc] (u1) -- (u3);
        \end{scope}
      \end{tikzpicture}\hfill\hfill
      \begin{tikzpicture}
        \path
        (0,-1.1) node {\phantom{$(d)$}}
        (-1,2) node[gu] (v1) {}
        (-1,1) node[gu] (v2) {}
        (-1,0) node[gu] (v3) {}
        (1,2) node[gv] (u1) {} edge (v2)
        (1,1) node[gv] (u2) {} edge (v2)
        (1,0) node[gu] (u3) {} edge (v3);
        \begin{scope}[on background layer]
          \draw[cc] (v1) -- (v3);
          \draw[cc] (u1) -- (u3);
        \end{scope}
      \end{tikzpicture}\hfill
      \begin{tikzpicture}
        \path
        (0,-1.1) node {\phantom{$(d)$}}
        (-1,2) node[gu] (v1) {}
        (-1,1) node[gu] (v2) {}
        (-1,0) node[gu] (v3) {}
        (1,2) node[gu] (u1) {} edge (v1)
        (1,1) node[gu] (u2) {} edge (v2) edge (v3)
        (1,0) node[gu] (u3) {} edge (v3);
        \begin{scope}[on background layer]
          \draw[cc] (v1) -- (v3);
          \draw[cc] (u1) -- (u3);
        \end{scope}
      \end{tikzpicture}\hfill
      \begin{tikzpicture}
        \path
        (0,-1.1) node {\phantom{$(d)$}}
        (-1,2) node[gu] (v1) {}
        (-1,1) node[gv] (v2) {}
        (-1,0) node[gv] (v3) {}
        (1,2) node[gu] (u1) {}
        (1,1) node[gv] (u2) {} edge (v2) edge (v3)
        (1,0) node[gv] (u3) {} edge (v2) edge (v3);
        \begin{scope}[on background layer]
          \draw[cc] (v1) -- (v3);
          \draw[cc] (u1) -- (u3);
        \end{scope}
      \end{tikzpicture}\hfill
      \begin{tikzpicture}
        \path
        (0,-1.1) node {\phantom{$(d)$}}
        (-1,2) node[gu] (v1) {}
        (-1,1) node[gu] (v2) {}
        (-1,0) node[gu] (v3) {}
        (1,2) node[gv] (u1) {} edge (v2)
        (1,1) node[gv] (u2) {} edge (v2)
        (1,0) node[gu] (u3) {} edge (v2) edge (v3);
        \begin{scope}[on background layer]
          \draw[cc] (v1) -- (v3);
          \draw[cc] (u1) -- (u3);
        \end{scope}
      \end{tikzpicture}\hfill
      \begin{tikzpicture}
        \path
        (0,-1.1) node {\phantom{$(d)$}}
        (-1,2) node[gu] (v1) {}
        (-1,1) node[gv] (v2) {}
        (-1,0) node[gv] (v3) {}
        (1,2) node[gv] (u1) {} edge (v1)
        (1,1) node[gv] (u2) {} edge (v1)
        (1,0) node[gu] (u3) {} edge (v2) edge (v3);
        \begin{scope}[on background layer]
          \draw[cc] (v1) -- (v3);
          \draw[cc] (u1) -- (u3);
        \end{scope}
      \end{tikzpicture}\hfill
      \begin{tikzpicture}
        \path
        (0,-1.1) node {\phantom{$(d)$}}
        (-1,2) node[gu] (v1) {}
        (-1,1) node[gu] (v2) {}
        (-1,0) node[gu] (v3) {}
        (1,2) node[gu] (u1) {} edge (v2)
        (1,1) node[gu] (u2) {} edge (v2) edge (v3)
        (1,0) node[gu] (u3) {} edge (v3);
        \begin{scope}[on background layer]
          \draw[cc] (v1) -- (v3);
          \draw[cc] (u1) -- (u3);
        \end{scope}
      \end{tikzpicture}
      \caption{Possible edge connections between color classes; vertices
        that belong to~$W$ because of these edges are depicted as a box; edges
        belonging to~$F$ are bold; the latter only appear in the pairs marked
        $(a)$, $(b)$, $(c)$ and $(d)$}
      \label{fig:connections}
    \end{figure}

    The following claim shows how the stable color partition of~$X$
    can be derived from the sets $W$~and~$F$ by a logspace
    computation.

    \begin{claim}
      On input~$X$, color refinement provides a unique color to a vertex
      $v\in C_i$ if and only if there is an $F$-path connecting~$v$ with
      some vertex in~$W$ or $v$~is the only vertex in its color class
      that is not reachable from~$W$ by an $F$-path.
    \end{claim}
    We prove the claim by induction on the number of rounds~$r$. We
    denote the length of a shortest $F$-path (if it exists) between a
    vertex~$v$ and some vertex in~$W$ by~$d(v,W)$. We show that the
    following equivalence holds for any $r\ge1$.
    \begin{quote}
      After round~$r$, vertex~$v$ has a unique color if and only if
      $d(v,W)<r$ or $v$~is the only vertex in its color class with
      $d(v,W)\ge r$.
    \end{quote}
    For $r=1$ the equivalence holds by definition of~$W$. Hence, it
    suffices to prove the equivalence for $r\ge2$ provided that it holds
    for $r-1$.  Let~$v$ be an arbitrary vertex. We first prove the
    backward implication of the equivalence.
    \begin{itemize}
     \item If $d(v,W)<r$, then there is an edge $\{v,w\}\in F$ with
      $d(w,W)<r-1$. By induction hypothesis, $w$~has a unique color after
      round $r-1$. But then also~$v$ has a unique color after round~$r$,
      since it is the unique neighbor of~$w$ in its color class (see
      Fig.~\ref{fig:connections}).
     \item If $v$~is the only vertex in its color class~$C_i$ with
      $d(v,W)\ge r$, it follows that $d(v',W)<r$ holds for all other
      vertices $v'\in C_i$. Hence, by using the same argument as above,
      it follows that all other vertices $v'\in C_i$ (and therefore all
      vertices in~$C_i$) have a unique color after round~$r$.
    \end{itemize}
    Next we prove the forward implication. We call two color classes
    \emph{linked}, if they are connected by at least one edge in~$F$
    (these pairs are marked as $(a)$, $(b)$, $(c)$ and~$(d)$ in
    Fig.~\ref{fig:connections}). By inspecting all unlinked pairs of
    color classes it is easy to verify that color refinements can only be
    propagated along linked color classes. Since $v$~receives a unique
    color in round~$r$ and since $v$~has to be distinguished from at most
    two other vertices in~$C_i$, either a single linked color class~$C_j$
    or at most two linked color classes $C_j$~and~$C_k$ cause the
    individualization of~$v$ in round~$r$. This means that at least one
    vertex in $C_j\setminus W$ has a unique color after round
    $r-1$. Hence, by induction hypothesis, one or more vertices
    $w_1,\dots,w_l\in C_j\setminus W$ are reachable from~$W$ by an
    $F$-path of length at most $r-2$. In the cases that $l\ge2$ or that
    $v$~is adjacent to some vertex in $\{w_1,\dots,w_l\}$ or that
    $C_i$~and~$C_j$ form a linked pair of type $(a)$, $(b)$ or $(c)$, it
    immediately follows that $d(v,W)\ge r$ holds for at most one vertex
    in~$C_i$.

    It remains to consider the case that the link between $C_i$~and~$C_j$
    is of type~$(d)$ and $v$~is not adjacent to the only vertex~$w_1$
    in~$C_j$ with $d(v,W)\le r-2$. Observe that in this case, the link
    between $C_i$~and~$C_j$ only causes the individualization of the
    neighbor~$v'$ of~$w_1$ in~$C_i$, but not the individualization of~$v$
    in round~$r$. Hence, there is a type~$(d)$ link between~$C_i$ and
    another color class~$C_k$ that causes the individualization of the
    third vertex $v''\in C_i$ in round~$r$. By the same argument as above
    it follows that $v''$~is adjacent to some vertex $w''\in C_k$ with
    $d(w,W)<r-1$. This concludes the proof of the claim and of the lemma
    since it follows that $v$~is the only vertex in~$C_i$ with
    $d(v,W)\ge r$.
  \end{proof}

  \begin{corollary}
    For any 3-bounded graph we can compute in logspace a vertex set~$S$
    of minimum size such that individualizing (or fixing) all the
    vertices in~$S$ makes the graph discrete, amenable, compact,
    refinable (or rigid).
  \end{corollary}

\subsection{Bounded number of refinement steps}

In this section, we consider (colored) graphs in which all color
classes become singletons after $\ell$~rounds of color refinement. We
denote the class of these graphs by $\ldiscrete$.

\begin{theorem}
  The $\IndLDiscrete$ problem is $\wtwo$-hard for any constant
  $\ell\ge 1$, even for uncolored and for $2$-bounded graphs.
\end{theorem}

 \begin{proof}
   We prove this by providing a reduction from the $\wtwo$-complete
   problem \DomSet.
   The input to this problem is a graph $X=(V,E)$ and
   a number~$k$ (treated as parameter) and the question is whether
   there exists a \emph{dominating set} $D\subseteq V$ of size~$k$ 
   in~$X$, meaning that each vertex $v\in V\setminus D$ is adjacent 
   to at least one vertex in~$D$. We transform the \DomSet instance 
   $(X,k)$ with $X=(V,E)$ into an equivalent instance $(X',k)$ where 
   $X=(V',E',c')$ for $\IndLDiscrete$.
   First we explain the construction using colors and afterwards we 
   show how to simulate the colors using a gadget. For this simulation
   it will be helpful if there is no vertex with degree zero in~$X$, 
   so if there are such vertices, we remove them in advance and 
   decrease $k$ accordingly.

   For every $v\in V$, the colored graph~$X'$ contains the vertices 
   $v_1,\dots,v_\ell$ and $v'_1,\dots,v'_\ell$ as well as the edges 
   $\{v_i,v_{i+1}\}$ and $\{v'_i,v'_{i+1}\}$ for all~$i$ in 
   $\{1,\dots,\ell-1\}$.  Furthermore, we add the edges $\{v_1,u_1\}$
   and $\{v'_1,u'_1\}$ for every edge $\{u,v\}$ of~$X$. We choose~$c'$
   in such a way that for all $v\in V$ the set $\{v_1, v'_1\}$ is a
   color class and $c'(v_i)=c'(v'_i)$ for all $i\in 
   \{2,\dots,\ell\}$. 

   Let~$D$ be a dominating set in~$X$.
   Individualizing all the vertices~$v_1$
   in~$X'$ with $v\in D$ will distinguish the pairs $\{v_1,v'_1\}$ for
   all $v\in V$ after one round of color refinement. Thus
   after at most $\ell-1$ more rounds all color classes of~$X'$ will be
   singletons.

   For the converse direction, let~$I$ be a set of 
   vertices in~$X'$, such that individualizing them and
   running $\ell$~rounds of color refinement produces singleton 
   color classes. If $I$~contains vertices $v_i$~or~$v'_i$ for $i>1$, 
   we can replace them by~$v_1$ and this still puts~$X'$ in 
   \ldiscrete.
   It is easy to see that this replacement does not decrease
   the number of color classes that become singletons after 
   $\ell$~rounds.
   Hence, we can assume that $I$~only contains vertices of the 
   form~$v_1$, implying that the set $D=\{v\in V\mid v_1\in I\}$ is a 
   dominating set of size at most~$\card{I}$ in~$X$. To see this it 
   suffices to observe that the vertices $u_\ell$~and~$u'_\ell$ 
   can only be distinguished by color refinement within 
   $\ell$~rounds if either $u_1$~is in~$I$ or $u$~has a 
   neighbor~$v$ for which $v_1$~is    in~$I$, implying that 
   either~$u$ or some neighbor of~$u$ is in~$D$.
   
   We now turn to the alternations to show the hardness for uncolored graphs
   and thus transform $(X',k)$ to $(X'',k'')$ for an uncolored graph 
   $X''=(V'',E'')$.
   Let~$n$ be the number of vertices in~$X$ and $h\colon V\to \{1,\dots,n\}$
   be an arbitrary numbering. We add the vertices $x_1,\dots,x_{n^2}$ as
   well as $y,y',z$ and~$z'$ to~$X''$. The edge set~$E''$ will further contain
   $\{x_i,x_j\}$ such that $i\neq j$ and $i+j\leq n^2+1$.
   Additionally, we connect
   each $v_1$~and~$v'_1$ to~$x_i$ if $i\leq h(v)n$.
   After this $\deg(v_i)=\deg(v'_1)\in \{h(v)n,\dots,(h(v)+1)n-1\}$
   (for $\ell=1$, else shifted by~$1$) for any $v\in V$. 
   For $i\leq\lfloor \frac{n^2}{2}\rfloor$
   we have $\deg(x_i)= n^2-i + 2n-2\lfloor \frac{i-1}{n}\rfloor$ and 
   $\deg(x_i)= n^2-i+1 + 2n-2\lfloor \frac{i-1}{n}\rfloor$ for 
   $i>\lfloor \frac{n^2}{2}\rfloor$.
   Thus, except for vertices $x_j$~and~$x_{j+1}$ with 
   $j=\lfloor\frac{n^2}{2}\rfloor$ the degree sequence among 
   the~$x_i$ is strictly decreasing. Since it is impossible to 
   construct a graph with at least two vertices and singleton degree 
   classes, we need some form of coloring (at least for $\ell=1$).
   To achieve this we connect $y$~and~$y'$ to all~$x_i$ vertices and 
   add the edges $\{z,z'\}$, $\{z,x_j\}$ and $\{z',x_j\}$.
   Since $y$~and~$y'$ and $z$~and~$z'$, respectively, have 
   the same neighborhood (we call such pairs twins), one of each pair 
   has to be individualized, otherwise $X''$~will
   not even become discrete.
   This comes at the price of setting $k''=k+2$, thus $(X'',k'')$ is
   our instance.
   
   Let $I\subseteq V'$ be some set such that $X'$~with all vertices
   in $I$ individualized has only singleton color classes after 
   $\ell$~rounds of color refinement. In~$X''$, we individualize 
   all the vertices in $I$ as well as $y$~and~$z$. After 
   individualization only the vertices $x_i$ have 
   $\deg_{\{y\}}(x_i)=1$ and for no vertex $u$ except $y'$
   $\deg(u)=n^2$ and $\deg_{\{y\}}(u)=0$ holds. Similarly, $z'$ and
   $x_j$ have a unique tuple of color degrees. Furthermore, only
   the vertices $v_i$ and $v'_i$ for $v\in V$ and $i>1$ may have
   a degree of at most 2 and be no neighbor of $z$ at the same time.
  
   For the converse direction, assume that we have individualized
   all vertices in some set $I$ in~$X''$ and all color 
   classes have become singletons after $\ell$~rounds. Further we
   assume that $I$ is chosen such that $\card{I}$ is minimal. 
   Then $\card{I\cap \{x,x',y'y'\}}=2$ must hold and $x_i\notin I$
   for all $i\in\{1,\dots,n^2\}$ since for all $v\in V$ the 
   vertices $v_1$ and $v'_1$ have the same neighbors among the $x_i$
   and we already have $\deg(u_1)\neq \deg(v_1)$ for $u\neq v$. Thus
   individualizing $I\setminus  \{x,x',y'y'\}$ puts $X'$ in
   \ldiscrete.
 \end{proof}
  The preceding proof is inspired by~\cite[Theorem~7]{R} describing
  an fixed parameter reduction from \DomSet to a problem called
  $d$-\prob{Distance Paired Dominating Set}, which asks for a given
  graph and a number~$k$ (treated as parameter) whether there is a
  set~$C$ of $k$~vertices such that all vertices in the graph are
  within distance~$d$ of a vertex in~$C$ and there is a perfect
  matching between the vertices in~$C$.

\section{The number of non-individualized vertices as parameter}\label{sec:nonindiv}

In this section, we show that the problem \nkdiscrete is in~\FPT. In
fact, we show a linear kernel and consequently, a
$k^{O(k)}n^{O(1)}$~time algorithm for this problem.

\begin{theorem}\label{nkdiscrete-fpt}
There exists a kernel of size~$2k$ for \nkdiscrete that can be computed in
polynomial time.
\end{theorem}
We begin with some notation. Given a colored graph $X=(V,E,c)$, let~$S$ be a
subset of vertices.  Let~$\stabcol[S]$ denote the stable partition
obtained by individualizing every vertex in $V \setminus S$ and
performing color refinement.  We denote the number of color classes
in~$\stabcol[S]$ by~$\card{\stabcol[S]}$.  We can partition the vertices~$u$
in $V \setminus S$ by their neighborhood $N(u)\cap S$ inside the set~$S$.
We denote this partition of $V \setminus S$ by~$\neigh[S]$ and
the number of sets in it by~$\card{\neigh[S]}$.  We call two vertices $u$~and~$v$
\emph{twins} if $N(u)\setminus \{v\} = N(v) \setminus \{u\}$.  This
relation is an equivalence relation and the corresponding equivalence
classes are called \emph{twin classes}.  A graph is said to be
\emph{twin-free} if each twin class is of size~$1$.

The following lemma shows that sufficiently large twin-free graphs are
\textsc{yes} instances of the \nkdiscrete problem.
\begin{lemma}\label{twin-free}
  Let $X=(V,E)$ be a twin-free graph. Suppose $\card{V}>2k$.  There exists
  a set $S \subset V$ of size~$k$ such that $\stabcol[S]$~is
  discrete.  Moreover, we can compute such a set in $(nk)^{O(1)}$~time.
\end{lemma}

\begin{proof}
  We describe the algorithm for computing~$S$. Initially, we pick an
  arbitrary subset $T_0\subset V$ of size~$k$ and run color refinement
  to compute the stable partition~$\stabcol[T_0]$. Let $C_1,\dots,C_l$
  be the color classes in~$\stabcol[T_0]$ that are contained in~$T_0$.
  If $\stabcol[T_0]$~is already discrete, we output the set $S = T_0$
  and stop.

  Otherwise we rename the color classes such that $\card{C_1}\ge\card{C_i}$ for
  $i=2,\dots,l$. Then we compute the partition
  $\neigh[S]=\{B_1,\dots,B_m\}$ of $V \setminus S$, where we assume
  that $\card{B_1}\ge\card{B_i}$ for $i=2,\dots,m$. If $m\geq k$, then we
  form~$S$ by picking an arbitrary vertex from each of the sets
  $B_1,\dots,B_k$. To see that $\stabcol[S]$~is discrete it suffices
  to observe that individualizing all the vertices in $V \setminus S$
  causes the separation of the sets $B_1,\dots,B_m$ and
  individualizing all but at most one vertex in each set~$B_i$ makes
  the graph discrete.

  It remains to handle the case that $m<k$. We show that in this case
  it is possible to compute in polynomial time a set~$T_1$ of size~$k$
  such that $\card{\stabcol[T_1]} > \card{\stabcol[T_0]}$. By 
  repeating this
  procedure $i\le k-1$ times, we end up with a set~$T_i$ for which
  $\stabcol[T_i]$~is discrete. Let $u$~and~$v$ be two vertices inside the
  color-class~$C_1$. Since $X$~is twin-free, there must be a vertex~$a$
  witnessing the fact that $u$~and~$v$ are not twins. Since $u$~and~$v$
  have the same color, $a$~cannot be individualized, implying that
  $a\in T_0$.  Let~$C_j$ be the color class containing~$a$. Since
  $C_1$~and~$C_j$ are stable color classes, there must exist a vertex
  $b \in C_j$ such that $\{u,a\}$ and $\{v,b\}$ are edges and
  $\{u,b\}$ and $\{v,a\}$ are non-edges. Clearly, individualizing~$a$
  refines the color class~$C_1$. Therefore, the set $T' = T_0- \{a\}$
  has the desired property $\card{\stabcol[T']} > \card{\stabcol[T_0]}$ but is
  of size $k-1$.

  Since $\card{V}>2k$ and $m<k$, it follows that $\card{B_1}\geq 2$. Let $x$~and~$y$
  be two vertices inside~$B_1$. Since $X$~is twin-free, there must be
  a vertex~$z$ witnessing the fact that $x$~and~$y$ are not
  twins. Since all vertices in~$T_0$ either have both vertices $x$~and~$y$
  as neighbors or none of them (otherwise, $x$~and~$y$ would have
  different neighborhoods inside~$T_0$, contradicting the fact that
  $x,y\in B_1$), it follows that $z\not\in T_0$.  We claim that the
  set $T_1 = T' \cup \{z\}$ yields the same stable partition as~$T'$,
  i.e., $\stabcol[T_1]=\stabcol[T']$. In fact, color refinement anyway
  assigns a unique color to~$z$, since it is the only
  non-individualized vertex that is adjacent to exactly one of the two
  individualized vertices $x$~and~$y$. This completes the proof of the
  lemma.
\end{proof}

\begin{proof}[Proof of Theorem~\ref{nkdiscrete-fpt}.]
  We now outline a simple kernelization algorithm for \nkdiscrete.
  Let~$X$ be the given graph and let~$k$ be the given parameter. The algorithm
  first makes the graph~$X$ twin-free by removing all but one vertex
  from each twin-class.  

  If the resulting graph~$X'$ has at most~$2k$ vertices, it outputs
  the instance $(X',k)$ as the kernel. Since in each twin class of~$X$,
  all but one vertices have to be individualized to make the
  graph discrete, the two instances $(X,k)$ and $(X',k)$ are indeed
  equivalent with respect to the \nkdiscrete problem.
 
  If $X'$~has more than~$2k$ vertices, the algorithm computes in
  polynomial time a set~$S$ of size~$k$ such that individualizing
  every vertex outside of~$S$ makes the graph~$X'$ discrete (see
  Lemma~\ref{twin-free}). Clearly this set~$S$ is also a solution
  for~$X$, so the kernelization algorithm can output a trivial
  \textsc{yes} instance.
\end{proof}

\end{document}